\documentclass[10pt,journal]{IEEEtranTCOM}
\usepackage{times,amsmath,amssymb,graphicx,epsfig,color,colortbl}
\usepackage{tikz}
\usepackage{caption}
\usepackage{soul}
\usepackage[normalem]{ulem}

\newtheorem{proposition}{Proposition}
\newtheorem{proof}{Proof}

\newtheorem{mydef}{Definition}

\def\argmin{\mathop{\arg\min}}

\def\mb#1{\mathbf{#1}}

\def\nn{\nonumber}
\def\beq{\begin{equation}}
\def\eeq{\end{equation}}
\def\beqa{\begin{eqnarray}}
\def\eeqa{\end{eqnarray}}
\def\ie{{\it i.e.,\ \/}}

\def\nn{\nonumber}
\def\lam{\lambda}
\def\beq{\begin{equation}}
\def\eeq{\end{equation}}
\def\beqa{\begin{eqnarray}}
\def\eeqa{\end{eqnarray}}

\def\bmtx{\begin{bmatrix}}
\def\emtx{\end{bmatrix}}

\begin{document}

\title{\bf MDS Codes with Progressive Engagement Property for Cloud Storage Systems}

\author{ Mahdi Hajiaghayi\IEEEauthorrefmark{1}
and Hamid Jafarkhani\IEEEauthorrefmark{1} \\
\IEEEauthorblockA{\IEEEauthorrefmark{1}Center of Pervasive Communications and
Computing, University of California, Irvine \\Email: \{mahdi.h,
hamidj\}@uci.edu}}

\maketitle
\thispagestyle{empty}
\pagestyle{empty}

\begin{abstract}
Fast and efficient failure recovery is a new challenge for cloud storage systems with a large number of storage nodes. A pivotal recovery metric upon the failure of a storage node
is \textit{repair bandwidth cost} 
which refers to the amount of data that must be downloaded for regenerating the lost data. Since all the surviving nodes are not always accessible, we intend to introduce a class of  maximum distance separable (MDS) codes that can be re-used when the number of selected nodes varies yet yields close to optimal repair bandwidth. Such codes  provide flexibility in engaging more surviving nodes in favor of reducing the repair bandwidth without redesigning the code structure and changing the content of the existing nodes.
We call this property of  MDS codes  \textit{progressive
engagement}. This name comes from the fact that if a failure occurs, it is shown that the best strategy is to incrementally engage the surviving nodes according to their accessing cost (delay, number of hops, traffic load or availability in general) until the repair-bandwidth or accessing cost constraints are met.
We argue that the existing MDS codes
fail to satisfy the progressive engagement property. We subsequently present a search algorithm to find a new set of codes named \textit{rotation codes} that has both progressive engagement and MDS properties. Furthermore, we illustrate how the existing permutation codes can provide progressive engagement by modifying the original recovery scheme.  Simulation results are presented
to compare the repair bandwidth  performance of such codes when the number of participating nodes varies as well as their speed of single failure recovery.

\end{abstract}

\section{Introduction}

With the advent of the cloud storage systems, ordinary
users and enterprizes are increasingly moving their data to the cloud in
order to
have higher reliability, availability, and accessibility. These massive data are
distributed across a large number of storage nodes via a cloud file system and
preserved for long time, perhaps forever. Due to various incidents, ranging from
regional disasters such as earthquake, to hardware/software update on storage
nodes, the cloud file system  must be able to reconstruct the user data only
from a subset of storage nodes. Such capability can be obtained through an
efficient fault-tolerant system based on replication and redundancy.

Mirroring is a basic yet popular solution that is usually adopted by local storage systems.  This solution which is also reflected in levels $1-2$ of Redundant
Array of Inexpensive Disk (RAID) \cite{lale1} simply keeps an exact copy of the data in multiple,
usually three, storage nodes. Mirroring results in high availability but it
comes at the price of high redundancy.

 While mirroring may work for a local storage system, its hardware and
operation cost (power, cooling systems and maintenance) drastically increases
with the size of data, thus not efficient for a cloud systems.  Therefore,
cloud file systems are transitioning to adopt erasure codes \cite{LAB17} , e.g. Windows
Azure Storage \cite{huang2012erasure} and Google GFS \cite{google_gfs}. Maximum distance separable (MDS) codes  \cite{DimakisSurvey2011} are optimal among erasure codes in terms of
redundancy-reliability trade-off . An $(n,k)$ MDS code when deployed into storage systems can tolerate up to $(n-k)$ node failures without data
loss while requiring a much lower redundancy compared with mirroring.  For an $(n,k)$ MDS code to be applied to a symmetric storage system, we must divide a data
of size $M$ into $k$ blocks and encode them into $n$ equal-size encoded blocks and store them at $n$ storage nodes. In this case, the whole data can be reconstructed  from any $k$ nodes out
of the total $n$ nodes. Coding in this context consists of two parts: First, encoding that determines the content of each node and second, decoding that specifies the recovery scheme when a node failure happens. Obviously, the decoding process highly depends on the encoding part.
Several MDS codes such as Reed-Solomon \cite{reed1960polynomial}, row-diagonal
parity (RDP) \cite{corbett2004row}, EVENODD \cite{blaum1995evenodd}, X-code
\cite{xu1999x}, and WEAVER codes \cite{hafner2005weaver} were proposed and applied to the storage systems  in the
literature  with specific features and decoding complexity to protect data
against multiple disk failures.

Other codes such as low-density parity-check (LDPC) codes and related codes such as
LT codes \cite{Luby02ltcodes} and Raptor codes \cite{shokrollahi2006raptor} have
also been considered for distributed storage systems \cite{plank2003practical}.
It is shown that while these codes exhibit a considerably low decoding
complexity due to their XOR-based structure, they fail to hold the MDS property. To address this problem, \cite{wylie2007} and
\cite{wylie2011} attempt to find some XOR-based MDS codes using an exhaustive
search approach over all possible cases. Due to the exponential growing search
space, they could only provide such codes for limited cases ( $k < 20$, and
$(n-k)< 10$). The authors of \cite{oggier2011self} aim at minimizing the
number of nodes to be contacted for recovery while offering MDS property for the
codes with rate less than $1/2$. Locally repairable codes (LRC) \cite{papailiopoulos2014locally} also considers the same metric as the code design metric. 

The amount of data that must be downloaded
from $d$ surviving nodes to repair a single failure is called \textit{repair
bandwidth} \cite{WuDimakis2009}.
When a permanent failure occurs, the lost data must be recovered and
regenerated at a new node to maintain the data integrity.  The downside of the traditional MDS codes is
that  the entire data must be downloaded before recovering the lost data.
Hence, its repair bandwidth is $M$, albeit to only $M/k$ portion of the
original
data being lost in a failed node. This high repair bandwidth incurs
considerable
file transfer and traffic on the network. This has sparked interest in finding
the minimum repair bandwidth which in many cases has led to a new class of MDS
codes
\cite{WuDimakis2009},\cite{Diamakis2010}, \cite{Shah2010}, \cite{Viveck2011},\cite{ViveckJafarMaleki2010} \cite{tamo2013zigzag}. Such codes have been proposed for some specific combinations of $(n,k,d)$ where $d$ is the number of surviving nodes participating in the recovery.
 All these newly proposed MDS codes require less repair bandwidth than $M$ and this bandwidth decreases as  more nodes are engaged for recovery. However, both encoding and decoding parts of these codes depend on $d$ and it must be known before designing the code. Therefore, for these MDS codes the coding structure and basically the content of nodes change as $d$ changes.
This is not practically appealing for two reasons. First, we may not know in advance how many nodes are available to participate for the recovery.
Second, when a new node is added to the storage system, the content of the existing nodes must be changed for the optimal repair bandwidth performance.

In this paper, we intend to introduce a family of  MDS codes that while requiring low repair bandwidth it can
be reused when the number of participating nodes varies without having to change
the entire code structure and the content of the nodes. This property which we call
\textit{progressive engagement}  provides flexibility in engaging
more surviving nodes in favor of reducing the repair bandwidth without
redesigning the code structure.
This name stems from a proven fact that if a failure occurs, the best strategy to recover the lost data in the presence of node access cost is to incrementally engage the surviving nodes according to their accessing cost (delay, number of hops, traffic load and availability in general) until the repair bandwidth or accessing cost constraints are met.
Progressive engagement property has another implication that allows for adding
new storage nodes to the system without changing the contents of the existing
nodes.

A similar approach to the progressive engagement can be seen in rate-compatible
punctured codes \cite{Error_control_book} for communication systems. Such codes
allow the transmission of information with different levels of protection, at
various rates without redesigning the code.

  We argue that the existing MDS codes fail to
satisfy this pivotal property. In traditional MDS codes such as Reed-Solomon
codes, engaging more nodes does not lead to the reduction of repair bandwidth.
On the other hand, progressive engagement is not satisfied in  the existing MDS
codes designed for reducing repair bandwidth \cite{rashmi2011optimal},
\cite{Shah2010},  \cite{tamo2013zigzag} since their
structure vary with the number of participating
nodes. This fact renders them non-practical for many systems with unknown number
of participating nodes in advance.

Our proposed coding solution differs from the regenerative codes \cite{Changho2010}  \cite{rashmi2011optimal} \cite{prakash2015storage} \cite{tian2013rate} \cite{rashmi2015having} as well. The structures of these codes also depend on the number of participating nodes for recovery.  The focus of regenerative codes are on a tradeoff  between the storage cost (amount of data is being stored in each node) and the repair bandwidth. In one side of this tradeoff, the \textit{Minimum Storage Regenerating} (MSR) codes seek the optimal repair bandwidth for minimum storage cost (storing $M/k$ amount of data in each node).  \cite{rashmi2015having} and \cite{rashmi2014hitchhiker} considered the problem of I/O overhead in addition to the repair bandwidth metric and proposed a solution that significantly reduces the I/O access. While the I/O access cost is not the focus of our paper, the proposed designs in these papers fail to work when the number of available nodes, \textit{helpers}, changes. 

 We  introduce a new class of MDS codes using a search algorithm that we call rotation codes with progressive engagement property. Moreover, we illustrate how the original recovery (decoding) scheme of the existing permutation code featured in \cite{Viveck2011} can be modified to accommodate the progressive engagement property.

The rest of the paper is organized as follows.
Section~\ref{system_model} describes the system model and state the problem of existing MDS codes with optimal repair bandwidth. The concept of progressive engagement for the MDS codes  is introduced in Section~\ref{progressive_engagement} along with its rigorous definition. Moreover, two examples of MDS codes with progressive engagement property are
introduced in this section. Numerical results and conclusions are provided in
Sections~\ref{numerical_sec} and
~\ref{conclusion_sec}, respectively.

\section{System Model and Problem Formulation} \label{system_model}
We consider a \textbf{symmetric}  storage system consisting of an array of total $n$ disks or
nodes\footnote{We use node and disk interchangeably throughout this paper.},
among which $k$ systematic nodes and $n-k=m$ parity nodes. Systematic nodes hold
the original data. Suppose a file of size $M$ is evenly partitioned into $k$
parts, each of size $L= \frac{M}{k}$ blocks, and stored in the $k$ systematic
nodes. Each block also consists of a fixed number of  $w$-bit symbols.
The content of parity nodes is calculated from the original data using an
erasure encoder. The coding and decoding operations are conducted over the Galois Field $2^{w}$, where $n< 2^w
-1$. The common values for $w$ are $8,16, 32$ depending on the size of the
network. We assume an erasure code $(n,k)$ with MDS property meaning
that any $k$ nodes suffice to recover the original data. This also means that up
to $(n-k)$ node failures can be tolerated without data loss. An example of such
system with $k=2$ systematic nodes S-node 1 and S-node 2, $m=2$ parity nodes
P-node 1 and P-node 2, and $L=2$ blocks is depicted in Fig.~\ref{exactrepair}.

 \begin{figure}[h]
 \centering
 \includegraphics[scale=.7]{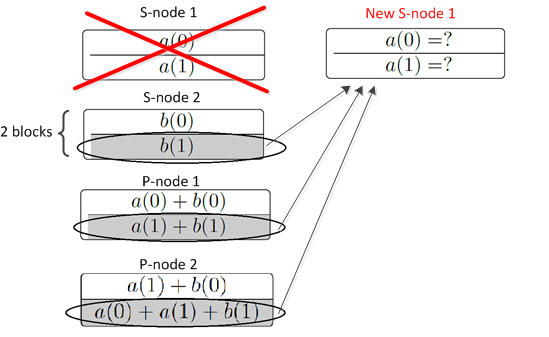}
 \caption{Exact repair for $n=4$, $k=2$, and $L=2$ blocks  ($M=4$ blocks)}
 \label{exactrepair}
 \end{figure}

To maintain the data integrity, when a node fails, the lost data must be
recovered and saved into a new node. For better exposition of this issue, assume
the first node fails in a system with $n=4$,  $k=2$ $L=2$ \cite{dimakis2011survey}. Using a traditional MDS
code, the new node could contact any 2 nodes and download $4$ blocks of data,
from which, $a(0)$ and $a(1)$ can be obtained. However, with the encoding scheme
shown in Fig. \ref{exactrepair}, it is also possible to do the following decoding to recover the lost data in Node 1: we can download only three
blocks $b(1)$; $a(1) + b(1)$; $a(0) +a(1) + b(1)$, from the three surviving
nodes, and attain $a(0)$ and $a(1)$ via solving a linear equation. In this case,
the repair bandwidth, defined as the minimum amount of data that must be
downloaded to recover the lost data, is $3$ blocks as opposed to $4$ for the
traditional MDS code.

Several work exist  that propose MDS codes  with
minimum repair bandwidth \cite{rashmi2011optimal}, \cite{Shah2010},
\cite{Viveck2011}, \cite{tian2013exact}, \cite{Changho2010},  \cite{rashmi2011optimal}, \cite{prakash2015storage}.
Let $d$ denote the number of surviving nodes participating
in the recovery process.
For these codes to work, $d$ is a parameter that must be known when designing a code. Moreover, most of these designs only exist  for a particular $d$ with $d= n-1$ being the most common case. In this case, if $d$ decrements by one (for example, a node becomes unavailable), in order to obtain the repair-bandwidth gain, a new code with $(n',k',d') = (n-1,k,d-1)$ must be loaded to the storage nodes, as if the unavailable node had not existed in the first place. The structures of these two codes can be quite different. To shed a light on this problem, consider the permutation code in \cite{Viveck2011} for the case $(n,k,d) = (5,2,4)$ in Table \ref{permutation_code}. If $d$ is set to $3$ in this code, the specific recovery scheme in \cite{Viveck2011} no longer works and we have to change the whole content of the parity nodes to that of the $(n',k',d')=(4,2,3)$ in Table \ref{permutation_code_4_2} to be able to recover the lost data. As can be seen, these two codes are completely different. Such a recalculation and reloading of the content cannot be tolerated in many practical cases. The parameters $\lam_i$'s in Tables  \ref{permutation_code} and  \ref{permutation_code_4_2}  are chosen so that the MDS properties are satisfied in these codes.

In this paper we consider a single failure scenario and assume that all the $k-1$ remaining systematic nodes participate for recovering the lost data. We denote the repair bandwidth by  $\gamma(p)$ where $p = d- (k-1)$ for a given coding scheme. It is noted that most of our definitions and proposed schemas are applicable to the general case where any $d$ nodes are selected for recovery.
\begin{table*}
\caption{Permutation code construction for $(n,k,d) = (5,2,4)$. The gray and green blocks are used in Section~\ref{pe_permutation} to describe the recovery scheme. }
\label{permutation_code}
\parbox{.11\linewidth}{
\centering
\caption*{S-node 1} 

\begin{tikzpicture}
\node (table) [inner sep=0pt] {
\begin{tabular}{c}
$a(0)$\\
\hline
$a(1)$ \\
\hline
$a(2)$ \\
\hline
$a(3)$ \\
\hline
$a(4) $\\
\hline
$a(5)$ \\
\hline
$a(6)$ \\
\hline
$a(7)$ \\
\hline
$a(8)$ \\
\end{tabular}
};
\draw [rounded corners=.2em] (table.north west) rectangle (table.south east);
\end{tikzpicture}
}
\hfill
\parbox{.11\linewidth}{
\centering
\caption*{S-node 2} 
\begin{tikzpicture}
\node (table) [inner sep=0pt] {
\begin{tabular}{c}
\cellcolor[gray]{.8} $b(0)$\\
\hline
\cellcolor[gray]{.8}$b(1)$ \\
\hline
\cellcolor[gray]{.8}$b(2)$ \\
\hline
\cellcolor[rgb]{0.5,1,0}$b(3)$ \\
\hline
\cellcolor[rgb]{0.5,1,0}$b(4)$ \\
\hline
\cellcolor[rgb]{0.5,1,0}$b(5)$ \\
\hline
$b(6) $\\
\hline
$b(7)$ \\
\hline
$b(8)$ \\
\end{tabular}
};
\draw [rounded corners=.2em] (table.north west) rectangle (table.south east);
\end{tikzpicture}
}
\hfill
\parbox{.15\linewidth}{
\centering
\caption*{P-node 1} 
\begin{tikzpicture}
\node (table) [inner sep=0pt] {
\begin{tabular}{c}
\cellcolor[gray]{.8}$\lam_1a(0) + \lam_2b(0)$\\
\hline
\cellcolor[gray]{.8}$\lam_1a(1) + \lam_2b(1)$\\
\hline
\cellcolor[gray]{.8}$\lam_1a(2) + \lam_2b(2)$\\
\hline
$\lam_1a(3) + \lam_2b(3)$\\
\hline
$\lam_1a(4) + \lam_2b(4)$\\
\hline
$\lam_1a(5) + \lam_2b(5)$\\
\hline
$\lam_1a(6) + \lam_2b(6)$\\
\hline
$\lam_1a(7) + \lam_2b(7)$\\
\hline
$\lam_1a(8) + \lam_2b(8)$\\
\end{tabular}
};
\draw [rounded corners=.2em] (table.north west) rectangle (table.south east);
\end{tikzpicture}
}
\hfill
\parbox{.15\linewidth}{
\centering
\caption*{P-node 2} 
\begin{tikzpicture}
\node (table) [inner sep=0pt] {
\begin{tabular}{c}
\cellcolor[gray]{.8}$\lam_1^2a(3) + \lam_2^2b(1)$\\
\hline
\cellcolor[gray]{.8}$\lam_1^2a(4) + \lam_2^2b(2)$\\
\hline
\cellcolor[gray]{.8}$\lam_1^2a(5) + \lam_2^2b(0)$\\
\hline
\cellcolor[rgb]{0.5,1,0}$\lam_1^2a(6) + \lam_2^2b(4)$\\
\hline
\cellcolor[rgb]{0.5,1,0}$\lam_1^2a(7) + \lam_2^2b(5)$\\
\hline
\cellcolor[rgb]{0.5,1,0}$\lam_1^2a(8) + \lam_2^2b(3)$\\
\hline
$\lam_1^2a(0) + \lam_2^2b(7)$\\
\hline
$\lam_1^2a(1) + \lam_2^2b(8)$\\
\hline
$\lam_1^2a(2) + \lam_2^2b(6)$\\
\end{tabular}
};
\draw [rounded corners=.2em] (table.north west) rectangle (table.south east);
\end{tikzpicture}
}
\hfill
\parbox{.15\linewidth}{
\centering
\caption*{P-node 3} 
\begin{tikzpicture}
\node (table) [inner sep=0pt] {
\begin{tabular}{c}
\cellcolor[gray]{.8}$\lam_1^3a(6) + \lam_2^3b(2)$\\
\hline
\cellcolor[gray]{.8}$\lam_1^3a(7) + \lam_2^3b(0)$\\
\hline
\cellcolor[gray]{.8}$\lam_1^3a(8) + \lam_2^3b(1)$\\
\hline
$\lam_1^3a(0) + \lam_2^3b(5)$\\
\hline
$\lam_1^3a(1) + \lam_2^3b(3)$\\
\hline
$\lam_1^3a(2) + \lam_2^3b(4)$\\
\hline
$\lam_1^3a(3) + \lam_2^3b(8)$\\
\hline
$\lam_1^3a(4) + \lam_2^3b(6)$\\
\hline
$\lam_1^3a(5) + \lam_2^3b(7)$\\
\end{tabular}
};
\draw [rounded corners=.2em] (table.north west) rectangle (table.south east);
\end{tikzpicture}
}
\end{table*}

\begin{table*}


\caption{Permutation code construction for $(n,k,d) = (4,2,3)$}
\label{permutation_code_4_2}
\parbox{.11\linewidth}{
\centering
\caption*{S-node 1} 

\begin{tikzpicture}
\node (table) [inner sep=0pt] {
\begin{tabular}{c}
$a(0)$\\
\hline
$a(1)$ \\
\hline
$a(2)$ \\
\hline
$a(3)$ \\
\end{tabular}
};
\draw [rounded corners=.2em] (table.north west) rectangle (table.south east);
\end{tikzpicture}
}
\hfill
\parbox{.11\linewidth}{
\centering
\caption*{S-node 2} 
\begin{tikzpicture}
\node (table) [inner sep=0pt] {
\begin{tabular}{c}
$b(0)$\\
\hline
$b(1)$ \\
\hline
$b(2)$ \\
\hline
$b(3)$ \\
\end{tabular}
};
\draw [rounded corners=.2em] (table.north west) rectangle (table.south east);
\end{tikzpicture}
}
\hfill
\parbox{.15\linewidth}{
\centering
\caption*{P-node 1} 
\begin{tikzpicture}
\node (table) [inner sep=0pt] {
\begin{tabular}{c}
$\lam_1a(0) + \lam_2b(0)$\\
\hline
$\lam_1a(1) + \lam_2b(1)$\\
\hline
$\lam_1a(2) + \lam_2b(2)$\\
\hline
$\lam_1a(3) + \lam_2b(3)$\\

\end{tabular}
};
\draw [rounded corners=.2em] (table.north west) rectangle (table.south east);
\end{tikzpicture}
}
\hfill
\parbox{.15\linewidth}{
\centering
\caption*{P-node 2} 
\begin{tikzpicture}
\node (table) [inner sep=0pt] {
\begin{tabular}{c}
$\lam_1^2a(2) + \lam_2^2b(1)$\\
\hline
$\lam_1^2a(3) + \lam_2^2b(0)$\\
\hline
$\lam_1^2a(0) + \lam_2^2b(3)$\\
\hline
$\lam_1^2a(1) + \lam_2^2b(2)$\\
\end{tabular}
};
\draw [rounded corners=.2em] (table.north west) rectangle (table.south east);
\end{tikzpicture}
}

\end{table*}

In the next section, we will introduce the concept of `progressive engagement' to address the problem stated above.

\section{Progressive Engagement Property}  \label{progressive_engagement}
We aim to design a family of MDS codes $(n,k,d)$ for $d=k,\cdots, n-1$ with the same encoding that while achieving low repair bandwidth it can be used when $d$ (or equivalently $p$) is not priori known. We call this property of MDS codes \textit{progressive engagement}.
An MDS code with progressive engagement  property provides flexibility in choosing the optimal set of participating nodes without having to change the coding scheme and the content of the storage nodes.
 We provide the following formal definition
for such codes.

\begin{mydef} \label{defi}
A family of MDS codes $(n,k,d)$, where $d=k,\cdots, n-1$, with $L$ blocks provides the progressive engagement property if
\begin{enumerate}
\item The encoder takes $k L$ blocks, from the $k$ systematic nodes, and provides $(n-k) L$ parity blocks for the $(n-k)$ parity nodes.
\item The decoder {\bf can} access $d L$ blocks from $d=k,\cdots, n-1$ surviving nodes and generate the $L$ missing blocks.
\item Let us define $p=d-k+1$, where $p=1,\cdots,n-k$. The decoder will only {\bf use} $\gamma(p)\leq d L$ blocks of data to generate the $L$ missing blocks.
\item For $p<p^\prime$, we have $\bar{\gamma}(p) < \bar{\gamma}(p^\prime)$ where $\bar{\gamma}$ is the average of repair bandwidth taken over all $p$ participating parity nodes for the recovery.
\end{enumerate}
\end{mydef}
For any value of $d$ (correspondingly  $p$ with the assumption that all systematic surviving nodes participate for recovery), we hope that such a code would require as minimum repair
bandwidth as if it were designed for that particular $d$. However, even if that
minimum bandwidth is not achieved for all $d$, the MDS codes with progressive
engagement property are still practically more amenable  over the codes that
only work for a fixed $d$.

While the progressive engagement property is practically appealing from many aspects as mentioned earlier, it is also motivated from the following proposition.

\begin{proposition}
Once a failure occurs, the best strategy to minimize a weighted sum of the  {\em accessing cost} and the repair bandwidth to recover the lost data is to incrementally engage the parity nodes  according to their accessing costs until both repair bandwidth and total accessing cost constraints are met.
\end{proposition}
\begin{proof}The proof is provided in Appendix \ref{optimization_section} \end{proof}

The advantage of the MDS codes with progressive engagement property over the
traditional MDS codes such as Reed-Solomon codes is the capability of reducing
repair bandwidth by involving more participating nodes. In fact, traditional MDS
codes fail to satisfy the second condition of Definition \ref{defi}, while the
recently proposed codes that yield minimum repair bandwidth fail to meet the
first condition. For the codes \cite{rashmi2011optimal}, \cite{Shah2010},
\cite{Viveck2011}, \cite{tian2013exact} the number of participating nodes must
be known when designing the codes. As $d$ may not be known a priori, it is a
great advantage of the MDS codes with progressive engagement property that the codes can
be used when $d$ varies without changing the code structure or the
content of the nodes.

 An analogy exists between the codes with progressive engagement property and rate-compatible punctured error
correcting codes \cite{Error_control_book}. The rate-compatible punctuated codes are very important in wireless communication since they allow transmitting information with various rates without redesigning the code while the code is as
efficient as if it were designed for the respective rate. Similarly, the progressive engagement property ensures to recover the lost data using any number of surviving nodes without redesigning the codes.

In what follows, we will provide two examples of MDS codes with progressive
engagement property.

\subsection{Progressive Engagement in Rotation Codes}

In this section, we propose a class of MDS codes that satisfies the progressive engagement property and design the corresponding encoding and decoding schemes.
We propose a computer search to find a code with the lowest repair bandwidth when different parity nodes participate in a single failure recovery. To make the computer search feasible, we reduce the search space by constraining the encoder to be a member of a family of codes that we call \textit{rotation codes}. A family of  $(n,k,d)$, $d=k,\cdots,n-1$, rotation codes with length $L$ is defined as follows. Since our proposed code is independent of $d$, we may skip $d$ in the code representation parameters. Let us define $\mb{a}_i,\ i=1,\cdots,k$ as the $k$ systematic node vectors. Then, the $n-k$ parity node vectors are defined as
\begin{align}
\mb{p_j} = \sum_{i=1}^{k} \lambda_{ij} \mb{a}_i \mb{R_{ij}}, \quad j=1,\cdots, n-k
\label{rotation_code}
\end{align}
where the coefficients $\lambda_{ij}$ are chosen to ensure the MDS property and $\mb{R_{ij}}$ is a rotation matrix defined below. Consider $\mb{I_L}=[\mb{e_1};\mb{e_2};\cdots;\mb{e_L}]$ as the $L\times L$ identity matrix, where $\mb{e_l}$ is a row vector of length $L$ with $1$ in the position $l$ and $0$ in every other positions.  Also, we define $\mb{R_{l}}=[\mb{e_{l+1}};\mb{e_{l+2}};\cdots;\mb{e_L};\cdots;\mb{e_{l}}]$ by cyclically rotating $I_L$'s rows $l$ times. Then, $\mb{R_{ij}}=\mb{R_{l}}$ for some $l$.

To search for a good code, one can try all possible $\mb{R_{ij}}$. Note that there are $k(n-k)$ positions for $\mb{R_{ij}}$ and the maximum complexity for a full search is $L^{k(n-k)}$ and finite. However, usually $\mb{R_{i1}}=\mb{R_{1j}}=\mb{I_L},\ i=1,\cdots,k,\ j=1,\cdots,n-k$ and one can use other symmetrical properties to further reduce the search space.
For a given set of rotation matrices $\mb{R_{ij}}$, we choose $\lambda_{ij}$ parameters to make the code an MDS code. This can be done by writing the MDS constraint equations and solving them for one set of $\lambda_{ij}$ parameters. The solution is not unique, but any choice of $\lambda_{ij}$ parameters as long as the code is an MDS code is acceptable.
The next step is to check if the code provides the progressive engagement property. To check the property, we consider all possible single failure cases and all possible involved parity nodes and calculate the corresponding repair bandwidth. If the repair bandwidth values are a decreasing function of the number of surviving nodes, i.e. involved parity nodes, the code is acceptable. In other words, for each code, we need to prove that in the case of single failure, the lost data can be re-constructed using any number of parity nodes.
Since the above rotation code is an XOR-based code, we can employ the algorithm proposed in \cite{khan2012rethinking} for single failure recovery while achieving low repair bandwidth. This algorithm aims to minimize the repair bandwidth regardless of $\lam_i$'s as long as they are non-zero. More specifically, for any systematic node failure and any participating set of surviving nodes, this algorithm constructs a directed weighted graph in which the shortest path between the root and any leaf (all leaves are connected) would determine the data blocks that must be downloaded and the minimum repair bandwidth. The height of this tree is equal to the number of blocks in the failed node. Using this algorithm, one can verify if both conditions of the progressive engagement property are satisfied for each choice of the code in \eqref{rotation_code}.
Note that the repair bandwidth resulted from this algorithm may not be the minimum but it is usually close to the minimum.
To describe the details, we concentrate on an example for $(n,k) = (6,3)$ rotation code with $L=4$ as a candidate:
\begin{align}
\mb{p_1} &= \mb{a} + \mb{b} + \mb{c} \nn \\
\mb{p_2} &= \mb{a} + \lambda_1\mb{R_1b} + \lambda_2\mb{R_3c} \nn \\
\mb{p_3} &= \mb{a} + \lambda_1^2\mb{R_2b} + \lambda_2^2\mb{R_1c}
\label{rotation_code63}
\end{align}
where
$\mb{R_1}= [\mb{e_2};\mb{e_3};\mb{e_4};\mb{e_1}], \mb{R_2} = [\mb{e_3}; \mb{e_4}; \mb{e_1}; \mb{e_2}]$ and $\mb{R_3} =[\mb{e_4};\mb{e_1};\mb{e_2};\mb{e_3}]$ are the rotation matrices with $\mb{e_i}$ being a row vector of length $4$ with $1$ in the position $i$ and $0$ in every other positions. The coefficients $\lam_1$ and $\lam_2$ in \eqref{rotation_code63} are chosen to ensure the MDS property. For  $\lam_1=2$ and $\lam_2 =3$  this property is explicitly proven in Appendix \ref{mds_property}.
The resulting structure is depicted in Table \ref{rotation_example}.
Next, we show that the code in Table \ref{rotation_example} provides the progressive engagement property.
As explained before, we use the algorithm proposed in \cite{khan2012rethinking}  to recover the single failure with minimum repair bandwidth.
The result of this algorithm is presented in Table \ref{Node1FailureRecovery} when S-node 1 fails and for all the sets of participating parity nodes.

\begin{table*}
\caption{Rotation code construction for $(n,k) = (6,3)$ and $L=4$ with progressive engagement property}
\label{rotation_example}
\parbox{.11\linewidth}{

\caption*{S-node 1} 
\begin{tikzpicture}
\node (table) [inner sep=0pt] {
\begin{tabular}{c}
$a(0)$\\
\hline
$a(1)$ \\
\hline
$a(2)$ \\
\hline
$a(3)$ \\

\end{tabular}
};
\draw [rounded corners=.2em] (table.north west) rectangle (table.south east);
\end{tikzpicture}
}
\hfill
\parbox{.11\linewidth}{
\caption*{S-node 2} 
\begin{tikzpicture}
\node (table) [inner sep=-0pt] {
\begin{tabular}{c}
$b(0)$\\
\hline
$b(1)$ \\
\hline
$b(2)$ \\
\hline
$b(3)$ \\
\end{tabular}
};
\draw [rounded corners=.2em] (table.north west) rectangle (table.south east);
\end{tikzpicture}
}
\hfill
\parbox{.11\linewidth}{

\caption*{S-node 3} 
\begin{tikzpicture}
\node (table) [inner sep=0pt] {
\begin{tabular}{c}
$c(0)$\\
\hline
$c(1)$ \\
\hline
$c(2)$ \\
\hline
$c(3)$ \\
\end{tabular}
};

\draw [rounded corners=.2em] (table.north west) rectangle (table.south east);
\end{tikzpicture}
}
\hfill
\parbox{.15\linewidth}{
\centering
\caption*{P-node 1} 
\begin{tikzpicture}
\node (table) [inner sep=0pt] {
\begin{tabular}{c}
$a(0) + b(0) +  c(0)$\\
\hline
$a(1) + b(1) +  c(1)$\\
\hline
$a(2) + b(2) +  c(2)$\\
\hline
$a(3) + b(3)+ c(3)$\\
\end{tabular}
};
\draw [rounded corners=.2em] (table.north west) rectangle (table.south east);
\end{tikzpicture}
}
\hfill
\parbox{.15\linewidth}{
\centering
\caption*{P-node 2} 
\begin{tikzpicture}
\node (table) [inner sep=0pt] {
\begin{tabular}{c}
$a(0) + 2b(1) + 3c(3)$\\
\hline
$a(1) + 2b(2) + 3 c(0)$\\
\hline
$a(2) + 2b(3) + 3 c(1)$\\
\hline
$a(3) + 2b(0) + 3c(2)$\\
\end{tabular}
};
\draw [rounded corners=.2em] (table.north west) rectangle (table.south east);
\end{tikzpicture}
}
\hfill
\parbox{.15\linewidth}{
\centering
\caption*{P-node 3} 
\begin{tikzpicture}
\node (table) [inner sep=0pt] {
\begin{tabular}{c}
$a(0) + 2^2b(2) + 3^2c(1)$\\
\hline
$a(1) + 2^2b(3) + 3^2c(2)$\\
\hline
$a(2) + 2^2b(0) + 3^2c(3)$\\
\hline
$a(3) + 2^2b(1) + 3^2c(0)$\\

\end{tabular}
};
\draw [rounded corners=.2em] (table.north west) rectangle (table.south east);
\end{tikzpicture}
}
\\

\end{table*}

\begin{table*}
\caption{Decoding strategy for a $(n,k) = (6,3)$ code with $L=4$ when node S-1 fails.}
\label{Node1FailureRecovery}
\begin{tabular}{|c|c|c|}
\hline
Accessible parity set & Optimal parity blocks downloaded & Repair bandwidth\\
\hline
P-node1 or P-node2 or P-node3 & all blocks in the selected parity node & $\gamma(p=1) = 12 $ \\
\hline
(P-node1,P-node2) & $p_2(3),p_2(1),p_1(2),p_1(0),c(0),c(2),b(0),b(2) $ & $\gamma(p=2) = 8 $ \\
\hline
(P-node1,P-node3) & $p_3(3),p_3(2),p_1(0),p_1(1),c(0),c(1),c(3),b(0),b(1) $ & $\gamma(p=2) = 9 $ \\
\hline
(P-node2,P-node3) & $p_2(3),p_2(2),p_2(1),p_2(0),c(0),c(3),b(0),b(1),b(2) $ & $\gamma(p=2) = 9 $ \\
\hline
(P-node1, P-node2,P-node3) & $p_2(3),p_2(1),p_1(2),p_1(0),c(0),c(2),b(0),b(2) $ & $\gamma(p=3) = 8 $ \\
\hline

\end{tabular}
\end{table*}
As can be seen in Table \ref{Node1FailureRecovery}, there are two sets with the same number of participating nodes that yield different repair bandwidths. Taking average over all failure and participating set cases, we have
\begin{align}
 \bar{\gamma} (p=1) = 12, \quad  \bar{\gamma} (p=2) = 8.66, \quad  \bar{\gamma} (p=3) = 8 \nn
\end{align}
It is evident that the rotation code in \eqref{rotation_code63} provides the progressive engagement property as the number of participating parity nodes increases, the repair bandwidth drops.

\subsection{Progressive Engagement in Permutation Code} \label{pe_permutation}

In this section, we consider the code structure based on the permutational
matrix featured in \cite{Viveck2011}. This code with its original recovery
scheme (decoding) cannot meet the progressive engagement property in Definition
\ref{defi}. However, we will describe how the recovery scheme can be modified
to allow the flexible engagement of parity
nodes.

\begin{proposition} \label{propos1}
The permutation code featured in \cite{Viveck2011} with the modified recovery
scheme described below provides the progressive engagement property. Furthermore, the repair bandwidth
cost for a given $p$ is
\begin{align}
&\gamma(p) = kL - \frac{L}{n-k} (p-1)(k-1), \label{repairbw}
\end{align}
where $L= (n-k)^k$.
\end{proposition}
\begin{proof}


We begin for the case  with $n=5$, $k=2$ and $L= 9$ to illustrate the basic idea of decoding that results in the progressive engagement prinaoperty.
 Following the permutation code structure in \cite{Viveck2011},  we reach to the code depicted in Table~
\ref{permutation_code} for $2$ systematic nodes and $3$ parity nodes.  $\lam_1$
and $\lam_2$ are two non-zero scalars in the field $\mathcal{F}(2^w)$ with $2^w
\geq 2k+1 = 5$ so that
\beq
\lam_1 + \lam_2 \neq 0 , \quad  \lam_1 \neq \lam_2.
\eeq
For simplicity, consider the case where S-node $1$ fails.  The original
recovery scheme only considers the case where all surviving nodes are engaged
to recover the single failure. In this case, the new node just needs to download
the first $\frac{1}{n-k}=\frac{1}{3}$ fraction of the total block rows of every
surviving node.  This fraction is indicated by the shaded gray region in Table \ref{permutation_code},
\ie S-nodes $2$ and P-nodes $1,2,3$.  The
downloaded blocks result in $12$ equations and $12$ unknowns, among which $9$
lost data $a(0),\cdots, a(8)$  can be readily obtained.
Hence,  the total repair-bandwidth $$\gamma(p=3) = (n-1)\times \frac{1}{n-k}  L
= 12 $$ suffices for recovering the lost data.  A similar recovery can be done
when S-node $2$ fails with the only difference that the block rows $l$ with $(l
\mod 3 =1)$ must be downloaded from every surviving node.

Next, we consider the case where only two parity nodes P-node $1$ and P-node
$2$ beside S-node $2$ are selected. In this case, the new node still downloads
the first three blocks of the available nodes S-node $2$ and P-nodes $1$ and $2$.
From these $9$ equations and $9$ unknowns, we can reconstruct the blocks of
vector $a$ appeared in the downloaded blocks, \ie $a(0), a(1), a(2)$ and
$a(3),a(4), a(5)$. The remaining blocks $a(6),a(7), a(8)$ can be obtained by
downloading the next three blocks of P-node $2$ and S-node $1$ (green region
in Table \ref{permutation_code}). This strategy results in $\gamma(p=2) = 3\times
\frac{L}{n-k} + 2\times \frac{L}{n-k}=15$ total repair bandwidth for $p=2$  or
$d= p + (k-1)=3$.  We note that while the achieved repair bandwidth may not be the minimum repair bandwidth possible, it is lower than the total
file size $M= 18$ blocks. We note that a similar strategy would also work if
the parity nodes other than P-node $1$ and P-node $2$ were selected. For
instance, for P-node $2$ and P-node $3$, in addition to the first $3$ rows of
these nodes and S-node $2$, by downloading the last three blocks of S-node $2$
and P-node $2$, we can reconstruct the whole $a(0), \cdots, a(8)$.
For $p=1$ we again download the first three blocks of that parity node along
with the first three blocks of S-node $2$, thus obtaining the respected
components of $a$ saved on those blocks. For the remaining $6$ blocks of $a$, we
can simply download the remaining $6$ blocks of S-node $2$ and the remaining $6$
blocks of the selected parity node, which in total results in $\gamma(p=1) = 2
\times \frac{L}{n-k} + 2 \times 3 \times \frac{L}{n-k} = 18$.

A similar argument can be made for the general case of the permutation codes. The proof of \eqref{repairbw}
for any permutation code with the modified recovery scheme is provided in Appendix \ref{pf_propos_App}.

\end{proof}
It remains unknown whether $\gamma(p)$ in \eqref{repairbw} is the optimal repair
bandwidth for the codes that support progressive engagement. The lower bound for
the exact repair with $d= p+ k-1$ participants is obtained in \cite{Diamakis2010}
\beq
\gamma_{min}(p) = \frac{L (p+k-1)}{p}. \label{lowerbound}
\eeq

In the simulation section, we provide a comparison between this lower bound and the repair bandwidth of our proposed codes.

\section{Simulation Results} \label{numerical_sec}
In this section, we provide some numerical experiments to assess the performance
of the rotation and  permutation codes.

As the first experiment, it would be insightful to compare the lower-bound \eqref{lowerbound} and the
repair bandwidth of the rotation code and permutation code with our modified recovery scheme in
\eqref{repairbw}. Fig. \ref{fig_repairBW} provides such a repair bandwidth
comparison and depicts the number of blocks versus the number of participating
parity nodes $p$ for  $(n,k) = (6,3)$. For completeness of
comparison, the repair bandwidth required by Reed-Solomon code is also included
in Fig. \ref{fig_repairBW}. Our proposed rotation code and the permutation code with our
modified recovery scheme significantly outperform the Reed-Solomon code in
terms of repair bandwidth saving. Furthermore, both of these codes satisfy the
progressive engagement property as their associated repair bandwidth drops as $p$
grows and do not need a new code structure and/or new node contents when p
changes. We also observe that our proposed rotation code results in lower repair bandwidth  compared to the permutation code with modified decoding when only two parity nodes are available for recovery. However, when all three parity nodes are involved, the permutation code requires the same bandwidth as the lower bound and lower than that of our proposed rotation code.
\begin{figure}
\centering
\includegraphics[scale=.65]{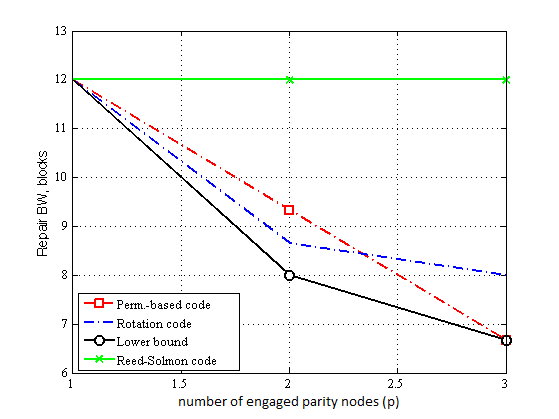}
\vspace{-3mm}
\caption{ Repair bandwidth versus the number of parity nodes participating in
recovery for the rotation code in \eqref{rotation_code} and permutation code with modified recovery scheme
and Reed-Solomon code $(n,k) =(4,2)$.} \label{fig_repairBW}
\vspace{-6mm}
\end{figure}

In the next experiment,  we consider
the $(n,k) = (10,3) $ permutation code and want to evaluate the speed of recovering the data of a failed node
versus the number of participating nodes for three symbol sizes $w=8,16, 32$. \footnote{This case is the default rate
for cloud file system Tahoe-LAFS \cite{Tahoe_Lafs}.}  We
implement this MDS code in C using the open source library for Galois Field
arithmetic in \cite{GF-complete}.
We assume an $M = 32$ MB file is distributed across $n=10$ storage nodes with $m=7$
parity nodes.  This follows $L = m^3 =343$ blocks of size roughly $\frac{M}{Lk}
= 32KB$. The speed here is defined as the amount of data recovered per second.
On our $1.8$ GHz Intel core i5, we obtain the results shown in Fig.
\ref{fig_speed_vs_p} using Monte Carlo simulation with $100$ iterations.  As
observed from this figure, the recovery speed increases as more parity nodes are engaged.
Moreover, although arithmetic calculation in $GF(2^{32})$ is more time-consuming
than $GF(2^8)$ and $GF(2^{16})$, the case corresponding to $w=32$ demonstrates
the highest speed among the three cases. This is due to the fact that for a fixed block size, the case of $w=32$ bits  has a lot less  symbols to process compared to the other cases.

In the previous experiment, the accessing cost of parity nodes was ignored in
favor of just showing the recovery speed of the code. In the next experiment, following the definition and terminologies of accessing costs in Appendix~\ref{optimization_section}, we
include the accessing cost in terms of the number of hops for a system whose $7$
parity nodes are located from $1$ to $7$ hops away, \ie $c_{1i}=i$.  For such a
system, Fig. \ref{fig_cost_vs_p} plots the total cost versus the number of
participating nodes for various coding schemes using two weights.  The proposed
solution in Section~\ref{system_model} readily obtains the optimal $p^*$ in our
code with progressive engagement property (denoted by Prog. Engag). This optimal
point clearly outperforms the Reed-Solomon code and the codes designed for
minimizing the repair bandwidth in terms of achieving a lower total cost.
Moreover, unlike our code and Reed-Solomon code, the codes designed for minimizing repair bandwidth  are just two
points in this plane instead of a curve and have no values for $p$ other than $p=(n-k) =7$.

\begin{figure}
\centering
\includegraphics[scale=.6]{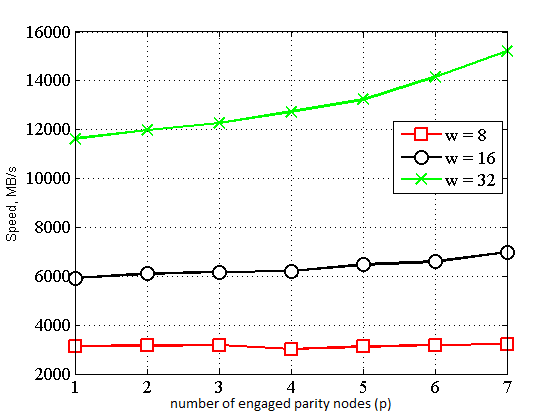}
\vspace{-3mm}
\caption{ CPU speed of single failure recovery versus the number of participating
nodes for permutation code with $(n,k) = (10,3)$} \label{fig_speed_vs_p}
\vspace{-7mm}
\end{figure}

\begin{figure}
\centering
\includegraphics[scale=.6]{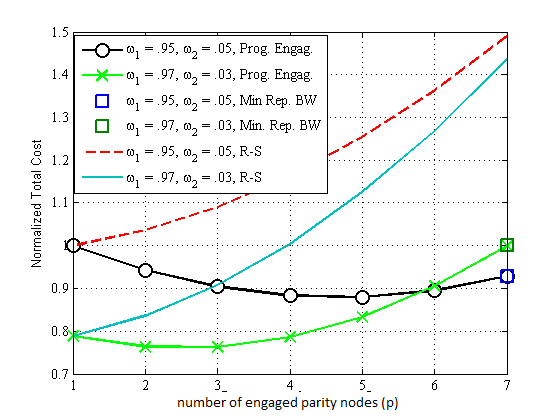}
\vspace{-2mm}
\caption{ Normalized total cost versus the number of participating nodes where
$(n,k) = (10,3)$ for various weights} \label{fig_cost_vs_p}
\vspace{-3mm}
\end{figure}

\section{Conclusions}  \label{conclusion_sec}
We introduced a new class of MDS
codes with a property called {progressive engagement}. This property provides
flexibility in engaging more surviving nodes in favor of reducing the repair
bandwidth without redesigning the code structure and content of the existing
nodes. We argued that such property is not met in the existing MDS codes. We further introduced a new class of MDS codes, called rotation codes, with progressive engagement property. Moreover, we illustrated how the existing permutation codes can provide progressive engagement by modifying the original recovery scheme. The repair bandwidth of the rotation codes and permutation codes were calculated and compared with other MDS codes.
\appendices

\section{Failure Recovery in the Presence of Accessing Cost} \label{optimization_section}
In this appendix, we provide the best strategy to recover the lost data of the failed node when both accessing cost and repair bandwidth matter.

It is not desirable to engage all the nodes in a cloud system
with a large number of geographically distributed nodes to recover the lost data. It is of our interest
to engage an optimal set of available nodes for the recovery process considering
their distances, link throughput as well as the total repair bandwidth.
Subsequently, we consider two specific cost metrics for an optimal selection
strategy; the accessing cost and the repair bandwidth cost.

The accessing cost is represented by the cost matrix $\mb{C} = [c_{ij}]_{n\times
n}$, where $c_{ij}$ denotes the cost associated with contacting and accessing
Node $j$ for recovering the data of Node $i$. The value of $c_{ij}$ can be a
function of multiple factors including communication delay, the number of hops
to reach Node $j$ and so on. The sum cost of all nodes of a set determines the
total accessing cost of that set.

Our objective is then to find $p^*$ nodes from the parity set P-nodes
$1,\cdots,m$ that minimize the weighted sum of the repair bandwidth and the
accessing cost. Note that the case where the objective is to minimize one of these
cost metrics subject to a constraint on the other metric can be converted to a weighted sum metric with Lagrange multipliers as the weights.
 To proceed, we define a binary indicator
$\quad  \alpha_i \in \{0,1\} ~, \quad \forall i=1,\cdots, m $
such that $\alpha_i$ is $1$ if P-node $i$ is selected, and $0$ otherwise.
Without loss of generality, assume that the systematic node S-node $1$ fails. Then, this
problem can be cast into the following integer programming problem
 \begin{align}
\min_{\alpha_{1},\cdots,\alpha_m} &\omega_1\sum_{i=1}^m \alpha_i c_{1i}  +
\omega_2 \gamma \left(\sum_{i=1}^m \alpha_i \right)\label{optimization} \\
s.t. \quad & I) \quad   \sum_{i=1}^m \alpha_i \geq 1   \\
& II) \quad  \alpha_i \in \{0,1\} ~, \quad i=1,\cdots, m, \nonumber
\end{align}
where the first and second terms in \eqref{optimization} encompass the
accessing cost and the repair bandwidth cost, respectively, $\omega_1$ and
$\omega_2$ are the respective weights normalized to $\omega_1+\omega_2 = 1$.
Condition $I$ in \eqref{optimization} requires at least one parity node to be
selected in order to ensure the reconstruction of the lost data for any MDS
code.

At the first glance, Problem \eqref{optimization} may seem to be hard to solve.
However, the solution becomes apparent by letting $\sum_{i=1}^m \alpha_i = p$.
Given $p$,  the optimization variables can be obtained from
 \begin{align}
A(p) = \min_{\alpha_{1},\cdots,\alpha_m} &w_1\sum_{i=1}^m \alpha_i c_{1i}
\label{optimization2} \\
s.t. \quad & I) \quad   \sum_{i=1}^m \alpha_i = p   \\
& II) \quad  \alpha_i \in \{0,1\} ~, \quad i=1,\cdots, m \nonumber
\end{align}
with solution
\beq
\alpha_i(p)^* =
\begin{cases}  1, & \quad \text{if } i \in \argmin_p (c_{11},\cdots, c_{1m}) \\
 0,  & \quad \text{otherwise}
\end{cases}~.
\label{Opt_alpha}
\eeq
where $\argmin_p (c_{11},\cdots, c_{1m})$ returns an index set of the first $p$
smallest $c_{11},\cdots, c_{1m}$ and can be easily done by sorting the accessing costs. 
Note that, so far, without any loss of generality, we have assumed the systematic node S-node $1$ has failed.
In order to avoid the above computation for each $p$, we
can just sort the accessing costs $c_{l1}, \cdots, c_{lm}$ for any failed node $l$ once with complexity $m\log
m$ using the quicksort algorithm and just return the index of the first $p$
elements.  Note that the $w_2 \gamma(p)$ was dropped from \eqref{optimization}
since, given $p$, it was a constant term. As $p$ can take only integer values
from $1$ to $m$, the solution to \eqref{optimization} can be obtained by
iterating over $p=1,\cdots, m$ and finding $p^*$ that achieves minimum $A(p) +
w_2 \gamma(p)$. Once $p^*$ is known, $\alpha_i(p^*) ^* $ determines whether
parity node $i$ is in the optimal selection set or not.  The overall search
complexity in this case is in the order of $O(m + m \log m)$.

The optimization procedure above leads us to a selection strategy that incrementally engages the parity
nodes according to their order of accessing cost.  After each new engagement, it
calculates the total cost of accessing and repair bandwidth.

\section{MDS Property of the Rotation Code in \eqref{rotation_code63}} \label{mds_property}
We need to prove that the coefficient matrix associated with any $3$ nodes selected from the total $6$ nodes is full rank. We consider the following cases:
\subsection*{Case 1: 3 Systematic Nodes Selected}
It is obvious that the coefficient matrix
associated with this selection is full rank.
\subsection*{Case 2: 2 Systemic Nodes and 1 Parity Node Selected}
This case is also trivial. For instance, consider S-node 1, S-node $2$ and P-node $3$ are selected. The rank of the coefficient matrix is
\begin{equation}
\text{rank} \left( \left[
\begin{matrix}
    \mb{I} & \mb{0}   & \mb{0} \\
    \mb{0}  & \mb{I}  & \mb{0} \\
    \mb{I} & 2^2\mb{R_2}  &3^2 \mb{R_3}
\end{matrix}
\right] \right) = Lk = 12
\end{equation}
which is full rank.
\subsection*{Case 3: $1$ Systemic Node and $2$ Parity Nodes selected}
First, let us verify the full rank coefficient matrix for a case where S-node $1$, P-node $2$ and P-node $3$ are selected. We need to show  that the determinant of the associated coefficient matrix is non-zero.
We have
\begin{align}
& \left|
\begin{matrix}
    \mb{I} & \mb{0}   & \mb{0} \\
    \mb{I}  & 2\mb{R_1}  & 3\mb{R_3} \\
    \mb{I} & 2^2\mb{R_2}  &3^2 \mb{R_1}
\end{matrix}
\right| =
\left|
\begin{matrix}
 2\mb{R_1}  & 3\mb{R_3} \\
2^2\mb{R_2}  &3^2 \mb{R_1}
\end{matrix}
\right|  =  \label{det_rotation_1}\\
& |2\mb{R_1}| \left| 3^2 \mb{R_1} - 2^2\mb{R_2}(2\mb{R_1})^{-1}3\mb{R_3} \right| = 6 | 3\mb{R_1} - 2\mb{I} | \neq 0  \label{det_rotation_2}
\end{align}
where from \eqref{det_rotation_1} to \eqref{det_rotation_2} we have used the following equation from the matrix theory
\beq
\left|
\begin{matrix}
    \mb{P} & \mb{S}   \\
    \mb{R}  & \mb{Q}
\end{matrix}
\right| = |\mb{P}|.|\mb{Q}-\mb{RP^{-1}S}|.
\eeq
Furthermore, equations $\mb{R_1}^{-1}= \mb{R_3}$, $\mb{R_2}\mb{R_3} = \mb{R_1}$  and $\mb{R_1}\mb{R_3} = \mb{I}$  were used to derive \eqref{det_rotation_2}.
A similar derivation is applicable to the other selection cases of $1$ systematic node and $2$ parity nodes.
\subsection*{Case 4: $3$ Parity Nodes Selected}
In this case, the determinant of the coefficient matrix will be
\begin{align}
& \left|
\begin{matrix}
    \mb{I} & \mb{I}   & \mb{I} \\
    \mb{I}  & 2\mb{R_1}  & 3\mb{R_3} \\
    \mb{I} & 2^2\mb{R_2}  &3^2 \mb{R_1}
\end{matrix}
\right| =
\left|
\begin{matrix}
 2\mb{R_1} -\mb{I}  & 3\mb{R_3} - \mb{I} \\
2^2\mb{R_2} - \mb{I} &3^2 \mb{R_1} - \mb{I}
\end{matrix}
\right|  =   |2\mb{R_1}-\mb{I}| \nn \\
& \left| 3^2 \mb{R_1} - \mb{I} - (2^2\mb{R_2} - \mb{I})(2\mb{R_1}-\mb{I})^{-1}(3\mb{R_3} - \mb{I}) \right| \neq 0 \label{det_rotation_3}
\end{align}
which can be easily verified to be non-zero.

\section{Proof of Proposition \ref{propos1} For General Case $(n,k)$ } \label{pf_propos_App}

We describe how the recovery scheme can be generalized to the case with
arbitrary $(n,k)$ and variable $p$. Without loss of generality, suppose
P-node $1$ fails. The recovery scheme involves two phases. In the first
phase, we download the same block rows of the given $p$ selected parity nodes and
the $k-1$ systematic nodes as the original recovery scheme would do with the
assumption of full participation. In total, we get $(k-1)\times \frac{L}{n-k} +
p \times \frac{L}{n-k}$ equations with the same number of unknowns, among which,
$p\times \frac{L}{n-k}$ components belong to Node $1$ and can be recovered. The
second phase is meant to recover the rest of the $((n-k) - p)\frac{L}{n-k}$
components of the parity Node $1$. For that, we consider any arbitrary parity node
(e.g. P-node $j$)  from the $p$ given parity nodes and all $k-1$ surviving
systematic nodes. We then download $((n-k) - p)\frac{L}{n-k}$ block rows from
P-node $j$  containing the components of P-node $1$ that were missing in the
first recovery phase. Next, we download all the components of these downloaded
vectors from their respective systematic nodes. This
would give us enough number of equations to recover the missing components of
P-nodes in the first phase. Hence the total repair bandwidth for a given $p$ is
\begin{align}
&\gamma(p)  = \big(p + (k -1) \big) \frac{L}{n-k} + k \bigg( \big((n-k) - p
\big) \bigg) \frac{L}{n-k} =  \nn\\
& \frac{-(p-1) (k-1) + k (n-k)}{n-k}L = kL - \frac{L}{n-k} (p-1)(k-1). \nn
\end{align}

\bibliographystyle{IEEEtran} 
\bibliography{IEEEabrv,ProgEng_2f}
\end{document}